\documentclass[letterpaper, 10 pt, conference]{ieeeconf} 
\IEEEoverridecommandlockouts  
\usepackage{amsmath}
\usepackage{amssymb}
\usepackage{amsthm}
\usepackage{mathtools}
\usepackage{derivative}
\usepackage{xcolor}
\usepackage{bm}

\usepackage[noadjust]{cite}
\usepackage{url}

\usepackage{graphicx}
\usepackage[export]{adjustbox} % For aligning figures
\graphicspath{{figures/}}

\newtheorem{theorem}{Theorem}
\newtheorem{lemma}{Lemma}
\newtheorem{proposition}{Proposition}

\theoremstyle{definition}
\newtheorem{definition}{Definition}

\newcommand{\argmin}{\operatornamewithlimits{arg\,min}}

\newcommand{\R}{\mathbb{R}}

\newcommand{\mc}[1]{\mathcal{#1}}
\newcommand{\T}{^\top}

\newcommand{\atantwo}{\operatorname{atan2}}
\newcommand{\bzero}{\mathbf{0}}

\newcommand{\bb}{\mathbf{b}}
\newcommand{\bc}{\mathbf{c}}
\newcommand{\bd}{\mathbf{d}}
\newcommand{\be}{\mathbf{e}}
\renewcommand{\bf}{\mathbf{f}} % NOTE: use \textbf if you really want to write bold non-math text.
\newcommand{\bg}{\mathbf{g}}

\newcommand{\bk}{\mathbf{k}}

\newcommand{\br}{\mathbf{r}}

\newcommand{\bu}{\mathbf{u}}
\newcommand{\bv}{\mathbf{v}}

\newcommand{\bx}{\mathbf{x}}
\newcommand{\by}{\mathbf{y}}
\newcommand{\bz}{\mathbf{z}}

\newcommand{\bB}{\mathbf{B}}

\newcommand{\bD}{\mathbf{D}}

\newcommand{\bG}{\mathbf{G}}
\newcommand{\bH}{\mathbf{H}}

\usepackage{xcolor}
\definecolor{myblue}{RGB}{49, 114, 174}
\definecolor{myred}{rgb}{0.796, 0.235, 0.2}
\definecolor{mygreen}{rgb}{0.22, 0.596, 0.149}
\definecolor{mypurple}{rgb}{0.584,0.345,0.698}

\usepackage{blindtext}

\title{\textbf{Compatibility of Multiple Control Barrier Functions \\ for Constrained Nonlinear Systems}}

\author{Max H. Cohen$^1$, Eugene Lavretsky$^2$, and Aaron D. Ames$^1$ %
\thanks{$^1$The authors are with the Department of Mechanical and Civil Engineering, California Institute of Technology, Pasadena, CA \texttt{\{maxcohen,ames\}@caltech.edu}.}
\thanks{$^2$The author is with The Boeing Company, Huntington Beach, CA \texttt{eugene.lavretsky@boeing.com}.}
\thanks{This research was supported by the Boeing Company. 
}
}

\begin{document}
\maketitle
\begin{abstract}
    Control barrier functions (CBFs) are a powerful tool for the constrained control of nonlinear systems; however, the majority of results in the literature focus on systems subject to a single CBF constraint, making it challenging to synthesize provably safe controllers that handle multiple state constraints. 
    This paper presents a framework for constrained control of nonlinear systems subject to box constraints on the systems' vector-valued outputs using multiple CBFs. Our results illustrate that when the output has a vector relative degree, the CBF constraints encoding these box constraints are compatible, and the resulting optimization-based controller is locally Lipschitz continuous and admits a closed-form expression. Additional results are presented to characterize the degradation of nominal tracking objectives in the presence of safety constraints. Simulations of a planar quadrotor are presented to demonstrate the efficacy of the proposed framework.
\end{abstract}

\section{Introduction}
Modern engineering systems -- from autonomous robots to aerospace systems -- are often subject to safety-critical operational requirements. From a controls perspective, these requirements manifest as state constraints where the primary control objective is to ensure that various states evolve in a desired operating region or set of the state space. The past decade has witnessed a surge in constrained control methodologies to address such problems, including, but not limited to, control barrier functions (CBFs) \cite{AmesTAC17}, model predictive control \cite{BorelliBemporadMorari}, reference governors \cite{KolmanovskyAutomatica17}, and Hamilton-Jacobi reachability \cite{BansalCDC17}. A comparison between many of these approaches can be found in \cite{AmesCSM23,FisacARCRAS23}. 

Among the aforementioned approaches, CBFs have gained popularity due to their ease of online implementation and their ability to instantiate \emph{safety filters} -- controllers that modify existing inputs in a minimally invasive fashion to ensure the satisfaction of state constraints \cite{GurrietICCPS18}. A major limitation of CBFs, however, is that the majority of results in the literature focus on systems subject to a single (scalar) state constraint. This is in sharp contrast to most real-world applications where systems are subject to complex safety requirements that are challenging to characterize with a single CBF. Despite the relative lack of theoretical results, using multiple CBFs has demonstrated much practical success \cite{WangTRO17}. 

Early attempts to formally characterize multiple CBFs focused on combining these CBFs into a single CBF candidate \cite{GlotfelterLCSS17,LarsLCSS19,TamasLCSS23}, but did not address the issue of \emph{compatibility}, that is, when these CBFs are mutually feasible. More recently, \cite{BlackCDC23} applied tools from adaptive control to combine multiple CBFs online, which are shown to be compatible under certain assumptions. Other works reason directly about the compatibility of multiple CBFs by focusing on special classes of systems and constraints such as single-input single-output (SISO) systems \cite{XuAutomatica18}, unicycles avoiding circular obstacles \cite{CortezLCSS22}, linear systems \cite{LavretskyACC25}, and fully actuated mechanical systems \cite{CortezAutomatica22}. Alternative approaches to addressing multiple CBFs involve computational techniques based on discretizing the state space \cite{TanCDC22}, viability kernel calculations \cite{BreedenACC23}, and sum of squares programming \cite{IsalyTAC24}. 

In this paper, we present a framework for constrained control of nonlinear systems with multiple state constraints using CBFs. We focus on nonlinear multi-input multi-output (MIMO) systems subject to box constraints on the system's output. Our main results illustrate that when this output has a vector relative degree, then i) the CBFs corresponding to each output constraint are compatible, ii) the resulting quadratic programming (QP)-based controller is locally Lipschitz continuous, and iii) this QP admits a closed-form solution. Furthermore, we characterize how these multi-CBF safety filters degrade nominal performance objectives.

In contrast to techniques that combine multiple CBFs into a single CBF \cite{GlotfelterLCSS17,LarsLCSS19,TamasLCSS23,BlackCDC23}, our framework directly addresses multiple CBFs, ensuring such CBFs are compatible by construction. Compared to related approaches that directly handle multiple CBFs \cite{XuAutomatica18,CortezLCSS22,LavretskyACC25,CortezAutomatica22}, our results apply to a fairly general class of systems that subsume those considered in, e.g., \cite{XuAutomatica18,LavretskyACC25,CortezAutomatica22}. 
Specifically, the present work extends the ideas introduced in \cite{LavretskyACC25}, which focused on LTI systems, to a more general class of square nonlinear systems.
Finally, unlike works such as \cite{TanCDC22,BreedenACC23,IsalyTAC24} our approach does not rely on computationally expensive procedures for verifying compatibility and thus has the potential to scale to higher dimensional systems. It should be noted that such computational approaches are often able to incorporate input bounds whereas our results do not currently account for such bounds. 

The remainder of this paper is organized as: Sec. \ref{sec:prelim} covers preliminaries on CBFs, relative degree, and multiple constraints; Sec. \ref{sec:main} and Sec. \ref{sec:stability} present our main results, which are illustrated via simulations in Sec. \ref{sec:sims}; Sec. \ref{sec:conclusions} concludes with a discussion on future research directions.

\section{Preliminaries}\label{sec:prelim}
\subsection{Notation}
We say that a function $h\,:\,\R^n\rightarrow\R$ is smooth if it is continuously differentiable as many times as necessary. We use $L_{\bf}h(\bx)\coloneqq \pdv{h}{\bx}(\bx)\bf(\bx)$ to denote the Lie derivative of a smooth function $h\,:\,\R^n\rightarrow\R$ along the vector field $\bf\,:\,\R^n\rightarrow\R^n$ with higher order Lie derivatives denoted as $L_{\bf}^ih(\bx)\coloneqq \pdv{L_{\bf}^{i-1}h}{\bx}(\bx)\bf(\bx)$; see \cite[Ch. 4]{Isidori}. We use $\|\bx\|$ to denote the Euclidean norm of a vector $\bx\in\R^n$ and $\|\bx\|_{\bG}\coloneqq\sqrt{\bx\T\bG\bx}$ to denote the weighted Euclidean norm with positive definite $\bG\in\R^{n\times n}$. We denote the cardinality of a set $\mc{A}$ as $|\mc{A}|$. Given $\bu\,:\,\R_{\geq0}\rightarrow\R^m$ we define $\|\bu(t)\|_{\infty}\coloneqq\sup_{t\geq0}\|\bu(t)\|$.

\subsection{Safety-Critical Control}
Consider a nonlinear control affine system with state $\bx\in\R^n$, input $\bu\in\R^m$, and dynamics:
\begin{equation}\label{eq:control-affine-system}
    \dot{\bx} = \bf(\bx) + \bg(\bx)\bu,
\end{equation}
where $\bf\,:\,\R^n\rightarrow\R^n$ models the system drift dynamics and the columns of $\bg\,:\,\R^n\rightarrow\R^{n\times m}$ capture the control directions. Unless otherwise stated, these functions are assumed to be sufficiently smooth on $\R^n$. The main objective of this paper is to design a locally Lipschitz feedback controller $\bk\,:\,\R^n\rightarrow\R^m$ that enforces state constraints, $\bx(t)\in\mc{C}\subset\R^n$ for all $t\geq0$, along trajectories of the closed-loop system. 
A powerful framework for enforcing state constraints on nonlinear systems is via CBFs \cite{AmesTAC17}. To this end, consider a smooth function $h\,:\,\R^n\rightarrow\R$ defining a state constraint set:
\begin{equation}\label{eq:C}
    \mc{C} \coloneqq \{\bx\in\R^n\,:\,h(\bx)\geq0\}.
\end{equation}
CBFs enable the construction of feedback controllers $\bk\,:\,\R^n\rightarrow\R^m$ that render the above set \emph{forward invariant}\footnote{A set $\mc{C}\subset\R^n$ is said to be forward invariant for \eqref{eq:fcl} if for each initial condition $\bx_0\in\mc{C}$, the resulting solution satisfies $\bx(t)\in\mc{C}$ for all $t\geq0$.} for the closed-loop system:
\begin{equation}\label{eq:fcl}
    \dot{\bx} = \bf_{\rm{cl}}(\bx) \coloneqq \bf(\bx) + \bg(\bx)\bk(\bx).
\end{equation}
\begin{definition}[\cite{AmesTAC17}]
    A smooth function $h\,:\,\R^n\rightarrow\R$ defining a set $\mc{C}\subset\R^n$ as in \eqref{eq:C} is said to be a CBF for \eqref{eq:control-affine-system} if there exists $\alpha>0$ such that for all $\bx\in\R^n$:
    \begin{equation}\label{eq:cbf}
        \sup_{\bu\in\R^m}\left\{ L_{\bf}h(\bx) + L_{\bg}h(\bx)\bu \right\} \geq - \alpha h(\bx).
    \end{equation}
\end{definition}
Note that the condition in \eqref{eq:cbf} is equivalent to \cite{jankovic2018robust}:
\begin{equation*}
    L_{\bg}h(\bx)=\bzero \implies L_{\bf}h(\bx) + \alpha h(\bx) \geq 0,
\end{equation*}
for all $\bx\in\R^n$. Yet it may be the case that $L_{\bg}h(\bx)\equiv\bzero$ (i.e., the relative degree of $h$ may be larger than one), and the above condition may not hold, motivating the use of higher order CBF approaches, such as the \emph{exponential} CBF \cite{SreenathACC16}.

\begin{definition}[\cite{SreenathACC16}]
    A smooth function $h\,:\,\R^n\rightarrow\R$ defining a set $\mc{C}\subset\R^n$ as in \eqref{eq:C} is said to be an \emph{exponential CBF} of order $r\in\mathbb{N}$ for \eqref{eq:control-affine-system} if there exists $\bm{\alpha}\in\R^{r}$ such that:
    \begin{equation}\label{eq:ECBF-poly}
        s^{r} + \alpha_rs^{r-1} + \dots + \alpha_2s + \alpha_1 = \prod_{i=1}^{r} (s_{i} - \nu_i),
    \end{equation}
    has negative real roots $(\nu_1,\dots,\nu_{r})\in\R^r$ and for all $\bx\in\R^n$:
    \begin{equation}\label{eq:ECBF}
        \sup_{\bu\in\R^m}\left\{L_{\bf}^rh(\bx) + L_{\bg}L_{\bf}^{r-1}h(\bx)\bu \right\} \geq - \bm{\alpha}\T\bH(\bx),
    \end{equation}
    where $\bH(\bx)\coloneqq(h(\bx),\,L_{\bf}h(\bx),\,\dots\,,\,L_{\bf}^{r-1}h(\bx))$.
\end{definition}
Given an Exponential CBF (ECBF), we recursively define:
\begin{equation}\label{eq:-psi-i}
    \begin{aligned}
        \psi_{0}(\bx) \coloneqq & h(\bx) \\
        \psi_{i}(\bx) \coloneqq & \odv{}{t}\psi_{i-1}(\bx) - \nu_{i}\psi_{i-1}(\bx),\quad \forall i\in\{1,\dots,r\},
    \end{aligned}
\end{equation}
and the associated sets:
\begin{equation}\label{eq:C=i}
    \mc{C}_i = \{\bx\in\R^n\,:\,\psi_{i}(\bx)\geq0\},\quad \forall i\in\{0,\dots,r\}.
\end{equation}
The following theorem captures the main result for ECBFs.
\begin{theorem}[\cite{SreenathACC16}]\label{thm:ECBF}
    If $h$ is an ECBF for \eqref{eq:control-affine-system}, then any locally Lipschitz feedback controller $\bk\,:\,\R^n\rightarrow\R^m$ satisfying:
    \begin{equation}\label{eq:ECBF-control}
        L_{\bf}^rh(\bx) + L_{\bg}L_{\bf}^{r-1}h(\bx)\bk(\bx) \geq - \bm{\alpha}\T\bH(\bx),
    \end{equation}
    for all $\bx\in\R^n$ renders the set $\mc{S}\coloneqq \cap_{i=0}^{r}\mc{C}_i$ forward invariant for the closed-loop system \eqref{eq:fcl}. Hence, if $\bx_0\in\mc{S}$, then $h(\bx(t))\geq0$ for all $t\geq0$.
\end{theorem}

The main utility of ECBFs is their ability to instantiate safety filters -- controllers that filter out unsafe actions from a desired controller $\bk_{\rm{d}}\,:\,\R^n\rightarrow\R^m$ to ensure the satisfaction of safety constraints. In particular, given an ECBF, one may leverage the optimization-based controller:
\begin{equation}
\begin{aligned}
    \bk(\bx) = \argmin_{\bu\in\R^m}\quad &\tfrac{1}{2}\|\bu - \bk_{\rm{d}}(\bx)\|_{\bG(\bx)}^2 \\ 
    \mathrm{s.t.} \quad & L_{\bf}^rh(\bx) + L_{\bg}L_{\bf}^{r-1}h(\bx)\bu \geq - \bm{\alpha}\T\bH(\bx)
\end{aligned}
\end{equation}
where $\bG\,:\,\R^n\rightarrow\R^{m\times m}$ is positive definite, to enforce forward invariance of $\mc{S}$ and, therefore, satisfaction of the constraint $h(\bx(t))\geq0$ for all $t\geq0$. 

\subsection{Outputs and Relative Degree}
In this paper, we focus on state constraints defined on a collection of \emph{outputs} $\by\,:\,\R^n\rightarrow\R^m$ of \eqref{eq:control-affine-system}.
Our ability to construct a controller enforcing such constraints relies on the notion of \emph{vector relative degree}, defined as follows.

\begin{definition}[\cite{Isidori}]
    A smooth function $\by\coloneqq(y_1(\bx),\dots,y_m(\bx))\in\R^m$ is said to have \emph{vector relative degree} $\br=(r_1,\dots,r_m)$ at $\bx\in\R^n$ if for all $i\in\{1,\dots,m\}$:
    \begin{equation*}
        L_{\bg}L_{\bf}^jy_i(\bx)=0,
    \end{equation*}
    for all $j\in\{0,\dots,r_i-2\}$ and the \emph{decoupling matrix}:
    \begin{equation}\label{eq:decoupling-matrix}
        \bB(\bx)\coloneqq 
        \begin{bmatrix}
            L_{\bg}L_{\bf}^{r_1-1}y_1(\bx) \\
             % L_{\bg}L_{\bf}^{r_2-1}y_2(\bx) \\
             \vdots \\ 
             L_{\bg}L_{\bf}^{r_m-1}y_m(\bx) \\
        \end{bmatrix}
        =
        \begin{bmatrix}
            \bb_1(\bx)\T \\
             % \bb_2(\bx)\T \\
             \vdots \\ 
             \bb_m(\bx)\T \\
        \end{bmatrix}
        \in\R^{m\times m},
    \end{equation}
    where $\bb_i(\bx)\coloneqq L_{\bg}L_{\bf}^{r_i-1}y_i(\bx)\T$, 
    is invertible.
\end{definition}

\subsection{Compatibility of Constraints}
The main topic addressed in this paper is the synthesis of controllers that enforce multiple CBF constraints simultaneously. To formalize such ideas, we first define what it means for multiple constraints to be compatible and then review necessary and sufficient conditions for compatibility.
\begin{definition}
    A collection of affine inequality constraints of the form $c_i(\bx) + \bd_i(\bx)\T\bu\geq 0$, where $c_i\,:\,\R^n\rightarrow\R$ and $\bd_i\,:\,\R^n\rightarrow\R^m$ for all $i\in\mc{I}\subset\mathbb{N}$, are said to be \emph{compatible} for \eqref{eq:control-affine-system} at $\bx\in\R^n$ if there exists an input $\bu\in\R^m$ such that:
    \begin{equation}
        c_i(\bx) + \bd_i(\bx)\T\bu \geq 0,\quad \forall i\in\mc{I}.
    \end{equation}
\end{definition}

\begin{lemma}\label{lemma:compatible}
     A collection of affine inequality constraints of the form $c_i(\bx) + \bd_i(\bx)\T\bu\geq 0$, where $c_i\,:\,\R^n\rightarrow\R$ and $\bd_i\,:\,\R^n\rightarrow\R^m$ for all $i\in\mc{I}\subset\mathbb{N}$, are compatible for \eqref{eq:control-affine-system} at $\bx\in\R^n$ if and only if for all $\lambda_i\geq0$:
     \begin{equation}\label{eq:compat-cond}
         \sum_{i\in\mc{I}}\lambda_i\bd_i(\bx)=\bzero \implies \sum_{i\in\mc{I}}\lambda_i c_i(\bx)\geq0.
     \end{equation}
\end{lemma}
\begin{proof}
    The proof follows that of \cite[Thm. 2]{TamasLCSS23}.
\end{proof}

\section{Compatibility of Multiple ECBFs}\label{sec:main}
The main objective of this paper is to synthesize controllers that enforce multiple state constraints on a nonlinear system of the form \eqref{eq:control-affine-system}. In this section, we provide conditions under which i) CBF constraints associated with these state constraints are compatible, ii) the resulting controller is well-defined, and iii) the controller can be expressed in closed-form. We address these challenges by focusing on the special, yet relevant, case of box constraints on the output $\by$ of \eqref{eq:control-affine-system}:
\begin{equation}\label{eq:output-constraints}
    \underbrace{
    \begin{bmatrix}
        \underline{y}_1 \\ \vdots \\ \underline{y}_m
    \end{bmatrix}}_{\underline{\by}}
    \leq
    \underbrace{
    \begin{bmatrix}
        y_1(\bx) \\ \vdots \\ y_m(\bx) \\
    \end{bmatrix}}_{\by(\bx)}
    \leq
    \underbrace{
    \begin{bmatrix}
        \overline{y}_1 \\ \vdots \\ \overline{y}_m
    \end{bmatrix}}_{\underline{\by}},
\end{equation}
where $\overline{y}_i > \underline{y}_i$ for all $i\in\{1,\dots,m\}$.
These output constraints can be encoded through $m$ pairs of ECBFs:
\begin{equation}\label{eq:ECBF-pairs}
    \begin{aligned}
        \underline{h}_{i}(\bx) = & y_i(\bx) - \underline{y}_i, \\
        \overline{h}_i(\bx) = & \overline{y}_i - y_i(\bx), \\ 
    \end{aligned}
\end{equation}
yielding a total of $2m$ ECBFs. Associated with each of these ECBFs is a state constraint set:
\begin{equation*}
    \begin{aligned}
        \underline{\mc{C}}^i \coloneqq & \{\bx\in\R^n\,:\, \underline{h}_i(\bx) \geq 0 \} = \{\bx\in\R^n\,:\,y_i(\bx) \geq \underline{y}_i\}, \\
        \overline{\mc{C}}^i \coloneqq & \{\bx\in\R^n\,:\, \overline{h}_i(\bx) \geq 0 \} = \{\bx\in\R^n\,:\,y_i(\bx) \leq \overline{y}_i\}, \\
    \end{aligned}
\end{equation*}
yielding the overall state constraint set:
\begin{equation}\label{eq:state-constraint-set}
    \mc{C} \coloneqq \left(\bigcap_{i=1}^{m}\underline{\mc{C}}^i \right) \bigcap \left(\bigcap_{i=1}^{m}\overline{\mc{C}}^i \right). 
\end{equation}
If $\by$ has vector relative degree $\br=(r_1,\dots,r_m)$ then $\overline{h}_i$ and $\underline{h}_i$ have relative degree $r_i$ and may be leveraged to construct a corresponding safe set:
\begin{equation}\label{eq:S-multi}
    \begin{aligned}
    \mc{S} \coloneqq  \left(\bigcap_{i=1}^m\underline{\mc{S}}_i \right) \bigcap \left(\bigcap_{i=1}^m\overline{\mc{S}}_i  \right)\subset\mc{C}, \\
        \underline{\mc{S}}_i \coloneqq  \bigcap_{j=0}^{r_i}\underline{\mc{C}}_j^i \quad
        \overline{\mc{S}}_i \coloneqq  \bigcap_{j=0}^{r_i}\overline{\mc{C}}_j^i \quad,
    \end{aligned}
\end{equation}
where $\underline{\mc{C}}_j^i$ and $\overline{\mc{C}}_j^i$ for $j\in\{0,\dots,r_i\}$ are defined as in \eqref{eq:-psi-i} and \eqref{eq:C=i}. Note that, based on the definition of $\mc{S}$ in \eqref{eq:S-multi}, $\underline{h}_i(\bx)\geq0$ and $\overline{h}_i(\bx)\geq0$ for all $\bx\in\mc{S}$, so that forward invariance of $\mc{S}$ implies satisfaction of the original state constraints. These ECBFs may be integrated into a single optimization-based safety filter:
\begin{equation}\label{eq:ecbf-qp-multi-1}
\begin{aligned}
    \min_{\bu\in\R^m}\quad &\tfrac{1}{2}\|\bu - \bk_{\rm{d}}(\bx)\|_{\bG(\bx)}^2 \\ 
    \mathrm{s.t.}
    \quad & L_{\bf}^{r_i}\underline{h}_i(\bx) + L_{\bg}L_{\bf}^{r_i-1}\underline{h}_i(\bx)\bu \geq - {\bm{\alpha}}_i\T\underline{\bH}_i(\bx) \\ 
    \quad & L_{\bf}^{r_i}\overline{h}_i(\bx) + L_{\bg}L_{\bf}^{r_i-1}\overline{h}_i(\bx)\bu \geq - {\bm{\alpha}}_i\T\overline{\bH}_i(\bx) \\ 
\end{aligned}
\end{equation}
for all $i\in\{1,\dots,m\}$, where $\bm{\alpha}_i\coloneqq(\alpha^i_1,\dots,\alpha^i_r)\in\R^r$ and:
\begin{equation*}
    \begin{aligned}
        \underline{\bH}_i(\bx) \coloneqq &
        \begin{bmatrix}
            \underline{h}_i(\bx) & L_{\bf}\underline{h}_i(\bx) & \hdots &  L_{\bf}^{r_i-1}\underline{h}_i(\bx)
        \end{bmatrix}\T
        \\
        = & \begin{bmatrix}
            y_i(\bx) - \underline{y}_i & L_{\bf}y_i(\bx) & \hdots &  L_{\bf}^{r_i-1}y_i(\bx) 
        \end{bmatrix}\T \\
        \overline{\bH}_i(\bx) \coloneqq &
        \begin{bmatrix}
            \overline{h}_i(\bx) & L_{\bf}\overline{h}_i(\bx) & \hdots &  L_{\bf}^{r_i-1}\overline{h}_i(\bx)
        \end{bmatrix}\T \\ 
        = & \begin{bmatrix}
            \overline{y}_i - y_i(\bx) & -L_{\bf}y_i(\bx) & \hdots &  -L_{\bf}^{r_i-1}y_i(\bx) 
        \end{bmatrix}\T. \\
    \end{aligned}
\end{equation*}
Moreover, noting that:
\begin{equation}
    \begin{aligned}
        L_{\bf}^{r_i}\underline{h}_i(\bx) = & L_{\bf}^ry_i(\bx) \\
        L_{\bf}^{r_i}\overline{h}_i(\bx) = & -L_{\bf}^ry_i(\bx) \\
        L_{\bg}L_{\bf}^{r_i-1}\underline{h}_i(\bx) = & L_{\bg}L_{\bf}^{r_i-1}y_i(\bx) \\ 
        L_{\bg}L_{\bf}^{r_i-1}\overline{h}_i(\bx) = & -L_{\bg}L_{\bf}^{r_i-1}y_i(\bx), \\ 
    \end{aligned}
\end{equation}
allows \eqref{eq:ecbf-qp-multi-1} to be equivalently written as:
\begin{equation}\label{eq:ecbf-qp-multi-2}
\begin{aligned}
    \min_{\bu\in\R^m}\quad &\tfrac{1}{2}\|\bu - \bk_{\rm{d}}(\bx)\|_{\bG(\bx)}^2 \\ 
    \mathrm{s.t.}
    \quad &  a_i(\bx) + \bb_i(\bx)\T\bu + \alpha_1^i(y_i(\bx) - \underline{y}_i) \geq 0 \\ 
    \quad & -a_i(\bx) - \bb_i(\bx)\T\bu - \alpha_1^i(y_i(\bx) - \overline{y}_i )  \geq 0, \\ 
\end{aligned}
\end{equation}
for all $i\in\{1,\dots,m\}$ where $\bb_i$ is defined in \eqref{eq:decoupling-matrix} and:
\begin{equation}\label{eq:ai-bi}
    \begin{aligned}
        a_i(\bx) \coloneqq & L_{\bf}^{r_i}y_i(\bx) + \sum_{j=1}^{r_i-1}\alpha_{j+1}^iL_{\bf}^jy_i(\bx).
    \end{aligned}
\end{equation}
The following lemma provides conditions under which the constraints from \eqref{eq:ecbf-qp-multi-1} are compatible for all $\bx\in\mathcal{S}$. 
\begin{lemma}\label{lemma:ECBF-compat}
    Let $\by\,:\,\R^n\rightarrow\R^m$ have relative degree $\br=(r_1,\dots,r_m)$ on an open set $\mathcal{E}\supset\mathcal{S}$ for \eqref{eq:control-affine-system} and consider the output constraints \eqref{eq:output-constraints}, the associated ECBFs \eqref{eq:ECBF-pairs}, and the resulting safety filter \eqref{eq:ecbf-qp-multi-1} (equivalently, \eqref{eq:ecbf-qp-multi-2}). Provided that $\alpha_1^i>0$ and $\overline{y}_i>\underline{y}_i$ for all $i\in\{1,\dots,m\}$, the constraints in \eqref{eq:ecbf-qp-multi-1}
    are compatible for all $\bx\in\mathcal{E}\supset\mathcal{S}$. 
\end{lemma}
\begin{proof}
    To show the constraints are compatible we leverage Lemma \ref{lemma:compatible}. To this end, define:
    \begin{equation}\label{eq:cj-dj}
        \begin{aligned}
            c_j(\bx) \coloneqq & a_j(\bx) + \alpha_1^j(y_j(\bx) - \underline{y}_j) \\ 
            c_{j+1}(\bx) \coloneqq & -a_j(\bx) -  \alpha_1^j(y_j(\bx) - \overline{y}_j )  \\ 
            \bd_{j}(\bx) = & \bb_{j}(\bx) \\ 
            \bd_{j+1}(\bx) = & -\bb_{j}(\bx), \\
        \end{aligned}
    \end{equation}
    for $j\in\{1,3,\dots,2m-1\}$. Now, we compute:
    \begin{equation*}
        \begin{aligned}
            \sum_{j=1}^{2m}\lambda_j\bd_j(\bx) = \sum_{i=1}^m(\underline{\lambda}_i - \overline{\lambda}_i)\bb_i(\bx) = \sum_{i=1}^p\mu_i\bb_i(\bx),
        \end{aligned}
    \end{equation*}
    where $\lambda_j=\underline{\lambda}_j$ and $\lambda_{j+1}=\overline{\lambda}_j$ for $j\in\{1,3,\dots,2m-1\}$, and $\mu_i\coloneqq \underline{\lambda}_i - \overline{\lambda}_i$.
    Under the relative degree assumption, the vectors $ \{\bb_i(\bx)\}_{i=1}^{m}$ are linearly independent. Hence: 
\begin{equation}
\begin{aligned}
    \sum_{i=1}^m\mu_i\bb_i(\bx) = \bzero \iff & \mu_i = 0,\;\forall i\in\{1,\dots,m\} \\
    \iff & \underline{\lambda}_i = \overline{\lambda}_i,\;\forall i\in\{1,\dots,m\}.
\end{aligned}
\end{equation}
Thus, when $\sum_{j=1}^{2m}\lambda_j\bd_{j}=\bzero$ we have:
\begin{equation*}
    \begin{aligned}
        \sum_{j=1}^{2m}\lambda_jc_j(\bx) = & \sum_{i=1}^m\underline{\lambda}_i\left[a_i(\bx) + \alpha_1^i(y_i(\bx) - \underline{y}_i) \right] \\ 
        & - \sum_{i=1}^m\overline{\lambda}_i\left[a_i(\bx) + \alpha_1^j(y_i(\bx) - \overline{y}_i )  \right] \\
        = & \sum_{i=1}^m\underline{\lambda}_i\alpha_1^i\left[\overline{y}_i - \underline{y}_i \right] \geq 0,
    \end{aligned}
\end{equation*}
for all $\underline{\lambda}_i=\overline{\lambda}_i\geq0$, which implies that \eqref{eq:compat-cond} holds. It thus follows from Lemma \ref{lemma:compatible} that the constraints in \eqref{eq:ecbf-qp-multi-1} are compatible for all $\bx\in\mathcal{E}\supset\mathcal{S}$, as desired.
\end{proof}

The preceding lemma states that the QP in \eqref{eq:ecbf-qp-multi-1} is always feasible: for any $\bx\in\mathcal{E}\supset\mathcal{S}$ there exists an input $\bu\in\R^m$ satisfying the collection of ECBF constraints simultaneously. Yet, this is not enough to ensure safety of the closed-loop system. One must additionally ensure that the resulting controller is at least locally Lipschitz continuous to guarantee the existence and uniqueness of solutions to the resulting closed-loop system. The following lemma illustrates that this controller is indeed locally Lipschitz continuous.

\begin{lemma}\label{lemma:lipschitz}
    Let the conditions of Lemma \ref{lemma:ECBF-compat} hold. Then, the optimal solution to \eqref{eq:ecbf-qp-multi-1}, $\bu^*\,:\,\mathcal{E}\rightarrow\R^m$, exists, is unique, and is locally Lipschitz continuous. 
\end{lemma}

\begin{proof}
    Since the objective function is strictly convex and radially unbounded, and the feasible set is non-empty (by Lemma \ref{lemma:ECBF-compat}) and closed, it follows from Weierstrass' Theorem that a unique optimal solution $\bu^*(\bx)$ exists for each $\bx\in\mathcal{E}$ \cite[Ch. 3.1.2]{Bertsekas}. We will say that constraint $j\in\{1,\dots,2m\}$ is \emph{active} at $\bx\in\mathcal{E}$ if the optimal solution satisfies:
    \begin{equation}
        c_{j}(\bx) + \bd_j(\bx)\bu^*(\bx) = 0,
    \end{equation}
    where $c_j$, $\bd_j$ are defined as in \eqref{eq:cj-dj}. For each $\bx\in\mathcal{E}$, let $\mc{A}(\bx)\subset\mathbb{N}$ denote the set of active inequality constraints, which we will show can have at most $m$ elements. For the sake of contradiction, suppose there exists a state $\bx\in\mathcal{E}$ such that $|\mc{A}(\bx)|=M>m$. This implies that there exists an $i\in\{1,\dots,m\}$ such that:
    \begin{equation*}
    \begin{aligned}
        a_i(\bx) + \bb_i(\bx)\T\bu^*(\bx) + \alpha_1^i(y_i(\bx) - \underline{y}_i) = & 0 \\ 
        - a_i(\bx) - \bb_i(\bx)\T\bu^*(\bx) - \alpha_1^i(y_i(\bx) - \overline{y}_i) = & 0.
    \end{aligned}
    \end{equation*}
    Yet, this implies that $\underline{y}_i= \overline{y}_i$, contradicting the assumption that $\underline{y}_i < \overline{y}_i$. Hence, we must have $|\mc{A}(\bx)|\leq m$. This also implies that only ``one side" of each output constraint can be active at a given $\bx\in\mathcal{E}$. For each $\bx\in\mathcal{E}$, consider the vectorized active constraints:
    \begin{equation*}
        \underbrace{
       \begin{bmatrix}
           c_{j_1}(\bx) \\ \vdots \\ c_{j_{|\mc{A}(\bx)|}}(\bx)
       \end{bmatrix}}_{\bc(\bx)} + 
       \underbrace{
       \begin{bmatrix}
           \bd_{j_1}(\bx)\T \\ \vdots \\ \bd_{j_{|\mc{A}(\bx)|}}(\bx)\T
       \end{bmatrix}}_{\bD(\bx)}
       = \bzero,
    \end{equation*}
    where $j_1,\dots,j_{|\mc{A}(\bx)|}\in\mc{A}(\bx)$. Since only one side of each output constraint can be active at a given $\bx$ and $\by$ has a relative degree, $\bD(\bx)$ has linearly independent rows for all $\bx\in\mathcal{E}$. Along with the assumption that all functions involved are locally Lipschitz continuous, this implies the conditions of \cite[Thm. 3.1]{HagerSIAM79} hold, and implies that the optimal solution $\bu^*\,:\,\mathcal{E}\rightarrow\R^m$ is locally Lipschitz continuous\footnote{See also \cite[Thm. 3]{DevLCSS23} and \cite[Thm. 1]{jankovic2018robust} for a similar argument.}.
\end{proof}

With the preceding results, we may now establish safety of the closed-loop system under the controller from \eqref{eq:ecbf-qp-multi-1}.

\begin{theorem}\label{thm:safety}
    Let the conditions of Lemma \ref{lemma:ECBF-compat} hold. Then, the controller $\bk(\bx)=\bu^*(\bx)$, with $\bu^*\,:\,\mathcal{E}\rightarrow\R^m$ the optimal solution to \eqref{eq:ecbf-qp-multi-1}, renders the set $\mc{S}$ from \eqref{eq:S-multi}
    forward invariant for the closed-loop system \eqref{eq:fcl}. Hence, for any initial condition $\bx_0\in\mc{S}$, the trajectory of \eqref{eq:fcl} satisfies:
    \begin{equation}\label{eq:output-constraint-2}
        \underline{\by}\leq\by(\bx(t))\leq \overline{\by},\quad \forall t\geq0.
    \end{equation}
\end{theorem}

\begin{proof}
    The controller $\bk(\bx)=\bu^*(\bx)$ from \eqref{eq:ecbf-qp-multi-1} ensures that:
    \begin{equation}
    \begin{aligned}
    L_{\bf}^{r_i}\underline{h}_i(\bx) + L_{\bg}L_{\bf}^{r_i-1}\underline{h}_i(\bx)\bk(\bx) \geq - {\bm{\alpha}}_i\T\underline{\bH}_i(\bx), \\ 
     L_{\bf}^{r_i}\overline{h}_i(\bx) + L_{\bg}L_{\bf}^{r_i-1}\overline{h}_i(\bx)\bk(\bx) \geq - {\bm{\alpha}}_i\T\overline{\bH}_i(\bx), \\ 
    \end{aligned}
    \end{equation}
    for all $i\in\{1,\dots,m\}$ and $\bx\in\mathcal{E}$. Since this controller is locally Lipschitz continuous (Lemma \ref{lemma:lipschitz}), it follows from Theorem \ref{thm:ECBF} that each $\underline{\mc{S}}_i$ and $\overline{\mc{S}}_i$ is forward invariant. As the intersection of forward invariant sets is forward invariant \cite[Prop. 4.13]{Blanchini}, it follows that $\mc{S}$ from \eqref{eq:S-multi} is forward invariant. Since $\bx\in\mc{S}$ also implies that $\underline{h}_i(\bx)\geq0$ and $\overline{h}_i(\bx)\geq0$, forward invariance of $\mc{S}$ ensures that \eqref{eq:output-constraint-2} holds. 
\end{proof}

The previous theorem formally establishes safety of the closed-loop system under the optimization-based controller \eqref{eq:ecbf-qp-multi-1}; however, in certain situations, it may be desirable to have access to an explicit, closed-form, expression of this controller. The following theorem illustrates that for particular choices of $\bG$, \eqref{eq:ecbf-qp-multi-1} admits a closed-form solution.

\begin{theorem}\label{thm:closed-form}
    Let the conditions of Lemma \ref{lemma:ECBF-compat} hold and suppose that $\bG(\bx)=\bB(\bx)\T\bB(\bx)$, where $\bB\,:\,\R^n\rightarrow\R^{m\times m}$ is the decoupling matrix from \eqref{eq:decoupling-matrix}. Then, the primal-dual solution $(\bu^*(\bx),\bm{\lambda}^*(\bx))$ to \eqref{eq:ecbf-qp-multi-1} is given by:
    \begin{equation}\label{eq:qp-multi-soln}
    \begin{aligned}
        \bu^*(\bx) = & \bk_{\rm{d}}(\bx) + \sum_{i=1}^m\big(\underline{\lambda}^*_i(\bx) - \overline{\lambda}^*_i(\bx)\big)\bG^{-1}(\bx)\bb_{i}(\bx), \\ 
        \underline{\lambda}^*_i(\bx) = & \max\{0, -\underline{\omega}_i(\bx)\}, \\
        \overline{\lambda}^*_i(\bx) = & \max\{0, -\overline{\omega}_i(\bx)\}, \\
        \underline{\omega}_i(\bx) \coloneqq & a_{i}(\bx) + \alpha_1^i(y_i(\bx) - \underline{y}_i) + \bb_{i}(\bx)\T\bk_{\rm{d}}(\bx), \\
        \overline{\omega}_i(\bx) \coloneqq & -a_{i}(\bx) - \alpha_1^i(y_i(\bx) - \overline{y}_i) - \bb_{i}(\bx)\T\bk_{\rm{d}}(\bx),
    \end{aligned}
    \end{equation}
    where $a_i$ is defined as in \eqref{eq:ai-bi} and $\bm{\lambda}^*=(\lambda_1^*,\lambda_2^*,\dots,\lambda_{2m}^*)$ with $\lambda_{j}^*=\underline{\lambda}_j^*$ and $\lambda_{j+1}^*=\overline{\lambda}_{j}^*$ for $j\in\{1,3,\dots,2m-1\}$. 
\end{theorem}

\begin{proof}
    Since \eqref{eq:ecbf-qp-multi-2} is a convex optimization problem with affine inequality constraints, and these affine inequality constraints are feasible by Lemma \ref{lemma:ECBF-compat}, Slater's condition holds \cite[Ch. 5.2.3]{Boyd}, which implies that the Karush-Kuhn-Tucker (KKT) conditions are necessary and sufficient for optimality \cite[Ch. 5.5.3]{Boyd}.
    The Lagrangian associated with \eqref{eq:ecbf-qp-multi-2} is:
    \begin{equation}\label{eq:lagrangian}
        L(\bx,\bu,\bm{\lambda}) = \tfrac{1}{2}\|\bu - \bk_{\rm{d}}(\bx)\|_{\bG(\bx)}^2 - \sum_{j=1}^{2m}\lambda_{j}\big[c_j(\bx) + \bd_{j}(\bx)\T\bu\big],
    \end{equation}
    where $c_j$ and $\bd_j$ are as in \eqref{eq:cj-dj} and the KKT conditions are:
    \begin{align}
        \pdv{L}{\bu}(\bx,\bu^*(\bx),\bm{\lambda}^*(\bx)) = 0,\quad & \label{eq:stationarity} \\
        c_{j}(\bx) + \bd_{j}(\bx)\T\bu^*(\bx) \geq 0,\quad & \forall j\in\mc{J} \label{eq:primal-feasibility} \\
        \lambda_j^*(\bx) \geq 0,\quad & \forall j\in\mc{J} \label{eq:dual-feasibility} \\
        \lambda_{j}^*(\bx)\left[ c_{j}(\bx) + \bd_{j}(\bx)\T\bu^*(\bx) \right] = 0,\quad & \forall j\in\mc{J} \label{eq:complimentary-slackness}
    \end{align}
    where $\mc{J}=\{1,\dots,2m\}$,
    corresponding to stationarity \eqref{eq:stationarity}, primal feasibility \eqref{eq:primal-feasibility}, dual feasibility \eqref{eq:dual-feasibility}, and complementary slackness \eqref{eq:complimentary-slackness}. From \eqref{eq:lagrangian} and \eqref{eq:stationarity}, we have:
    \begin{equation*}
        \begin{aligned}
            \pdv{L}{\bu}(\bx,\bu,\bm{\lambda}) = & (\bu - \bk_{\rm{d}}(\bx))\bG(\bx) - \sum_{j=1}^{2m}\lambda_{j}\bd_{j}(\bx)\T \\
            = & (\bu - \bk_{\rm{d}}(\bx))\bG(\bx) - \sum_{i=1}^m\big(\underline{\lambda}_{i} - \overline{\lambda}_i \big)\bb_{i}(\bx)\T,
        \end{aligned}
    \end{equation*}
    from which one may verify that, with $\bu=\bu^*(\bx)$ and $\bm{\lambda}=\bm{\lambda}^*(\bx)$ as in \eqref{eq:qp-multi-soln}, the stationarity condition \eqref{eq:stationarity} is satisfied\footnote{Note that $\bG$ is invertible since $\bB$ has full rank based on the relative degree assumption.}. 
    Moreover, the candidate dual solution $\bm{\lambda}^*$ from \eqref{eq:qp-multi-soln} satisfies the dual feasibility condition \eqref{eq:dual-feasibility} by construction. Using \eqref{eq:cj-dj} and \eqref{eq:qp-multi-soln}, primal feasibility \eqref{eq:primal-feasibility} then dictates that:
    \begin{equation*}
        \begin{aligned}
            \underline{\omega}_i(\bx) + \sum_{k=1}^m\big(\underline{\lambda}^*_k(\bx) - \overline{\lambda}^*_k(\bx)\big)\bb_{i}(\bx)\T\bG^{-1}(\bx)\bb_{k}(\bx) \geq 0,  \\
            \overline{\omega}_i(\bx) - \sum_{k=1}^m\big(\underline{\lambda}^*_k(\bx) - \overline{\lambda}^*_k(\bx)\big)\bb_{i}(\bx)\T\bG^{-1}(\bx)\bb_{k}(\bx) \geq 0,
        \end{aligned}
    \end{equation*}
    for all $i\in\{1,\dots,m\}$.  Noting that $\bG(\bx)=\bB(\bx)\T\bB(\bx)$ is the Gram matrix associated with $\{\bb_i(\bx)\}_{i=1}^m$, we have\footnote{This follows from the element-wise definition of matrix multiplication.}:
    \begin{equation*}
        \bb_i(\bx)\T\bG^{-1}(\bx)\bb_{j}(\bx) = \begin{cases}
            1, & \text{if } i = j, \\ 
            0, & \text{if } i\neq j,
        \end{cases}
    \end{equation*}
    so that the primal feasibility condition reduces to:
    \begin{equation}\label{eq:primal-feasibility-simple}
        \begin{aligned}
           \underline{\omega}_i(\bx) + \big(\underline{\lambda}^*_i(\bx) - \overline{\lambda}^*_i(\bx)\big) \geq 0, \\ 
            \overline{\omega}_i(\bx)  -  \big(\underline{\lambda}^*_i(\bx) - \overline{\lambda}^*_i(\bx)\big) \geq 0, 
        \end{aligned}
    \end{equation}
    for all $i\in\{1,\dots,m\}$. To show that \eqref{eq:qp-multi-soln} ensures primal feasibility, we note that the left-hand-side of the first inequality in \eqref{eq:primal-feasibility} can be expressed as:
    \begin{equation*}
        \mathrm{LHS}(\bx) \coloneqq\underline{\omega}_i(\bx) + \max\{0, -\underline{\omega}_i(\bx)\} - \max\{0,-\overline{\omega}_i(\bx) \},
    \end{equation*}
    whose behavior we will analyze by considering four separate cases: i) $\underline{\omega}_i(\bx) \geq0$ and $\overline{\omega}_i(\bx) \geq0$; ii) $\underline{\omega}_i(\bx) \leq0$ and $\overline{\omega}_i(\bx) \geq0$; iii) $\underline{\omega}_i(\bx) \geq0$ and $\overline{\omega}_i(\bx) \leq0$; iv) $\underline{\omega}_i(\bx) \leq0$ and $\overline{\omega}_i(\bx) \leq0$. For cases (i)-(iii) we have $\mathrm{LHS}(\bx)\geq0$ since:
    \begin{equation*}
    \begin{aligned}
        \text{Case i):} \quad & \mathrm{LHS}(\bx) = \underline{\omega}_i(\bx) \geq 0, \\ 
        \text{Case ii):} \quad & \mathrm{LHS}(\bx) = \underline{\omega}_i(\bx) - \underline{\omega}_i(\bx) = 0, \\ 
        \text{Case iii):} \quad & \mathrm{LHS}(\bx) = \underline{\omega}_i(\bx) + \overline{\omega}_i(\bx) = \alpha_1^i\big(\overline{y}_i - \underline{y}_i\big) > 0. \\ 
    \end{aligned}
    \end{equation*}
    Next, we claim that there exists no $\bx\in\mathcal{E}$ for which case (iv) holds. Indeed, if both $\underline{\omega}_i(\bx) \leq0$ and $\overline{\omega}_i(\bx) \leq0$ then:
    \begin{equation*}
        \begin{aligned}
            a_i(\bx) + \bb_i(\bx)\T\bk_{\rm{d}}(\bx) + \alpha_1^i y_i(\bx) \leq \alpha_1^i \underline{y}_i, \\ 
            a_i(\bx) + \bb_i(\bx)\T\bk_{\rm{d}}(\bx) + \alpha_1^i y_i(\bx) \geq \alpha_1^i \overline{y}_i,
        \end{aligned}
    \end{equation*}
    which implies that $\overline{y}_i\leq \underline{y}_i$, contradicting the assumption that $\overline{y}_i > \underline{y}_i$. Thus, case (iv) never occurs. Following analogous steps for the second inequality in \eqref{eq:primal-feasibility-simple}, the preceding argument demonstrates that \eqref{eq:qp-multi-soln} satisfies the primal feasibility requirement \eqref{eq:primal-feasibility}. It thus remains to show that \eqref{eq:qp-multi-soln} satisfies the complimentary slackness condition \eqref{eq:complimentary-slackness}. Using \eqref{eq:ai-bi}, \eqref{eq:qp-multi-soln}, and \eqref{eq:primal-feasibility-simple}, this condition may be expressed as:
    \begin{equation*}
    \begin{aligned}
         \max\{0, -\underline{\omega}_i(\bx)\}\big[ \underline{\omega}_i(\bx) + \max\{0, -\underline{\omega}_i(\bx)\} - \overline{\lambda}^*_i(\bx)\big] = 0 ,\\
         \max\{0, -\overline{\omega}_i(\bx)\}\big[\overline{\omega}_i(\bx) - \underline{\lambda}_i^*(\bx) + \max\{0, -\overline{\omega}_i(\bx)\} \big] = 0,
    \end{aligned}
    \end{equation*}
    for all $i\in\{1,\dots,m\}$. One may verify that the above holds for cases (i)-(iii) outlined above and, by noting that case (iv) never occurs, we have that \eqref{eq:qp-multi-soln} satisfies the complementary slackness requirement \eqref{eq:complimentary-slackness}. Thus, \eqref{eq:qp-multi-soln} satisfies the KKT conditions and is therefore the optimal solution to \eqref{eq:ecbf-qp-multi-1}. 
\end{proof}

Theorem \ref{thm:closed-form} demonstrates that by selecting the cost function appropriately, it is possible to derive a closed-form solution to \eqref{eq:ecbf-qp-multi-1}. This was accomplished by selecting the weight $\bG$ in the quadratic cost as the Gram matrix associated with the decoupling matrix $\bB$, effectively decoupling the constraints in \eqref{eq:ecbf-qp-multi-1}, allowing the associated Lagrange multipliers to be considered independently. The motivation for deriving this closed-form solution stems from the use of CBFs in applications where running an optimizer in the control loop may be impractical or unacceptable. Examples include resource-constrained platforms where the computational overhead of instantiating an optimization solver may be infeasible, and in aerospace systems where various requirements often prohibit the use of optimizers in the control loop. 

\section{Stability Analysis}\label{sec:stability}
In this section, we study the stability properties of \eqref{eq:control-affine-system} equipped with the safety filters developed in Sec. \ref{sec:main} in the context of trajectory tracking. To this end, let $\by_{\rm{d}}\,:\,\R_{\geq0}\rightarrow\R^n$ be a smooth desired output trajectory and define:
\begin{equation}\label{eq:tracking-error}
    \be(\bx, t) \coloneqq \by(\bx) - \by_{\rm{d}}(t),
\end{equation}
as the tracking error. To make further statements regarding stability, we must, at a bare minimum, ensure that \eqref{eq:tracking-error} is bounded when the system is controlled by a safety filter.
\begin{lemma}
    Let the conditions of Theorem \ref{thm:safety} hold and suppose that $\by_{\rm{d}}(t)$ is bounded for all $t\geq0$. Then, along trajectories of the closed-loop system \eqref{eq:control-affine-system} with $\bu=\bk(\bx)$ in \eqref{eq:ecbf-qp-multi-1}, the tracking error $\be(\bx(t),t)$ is bounded for all $t\geq 0$.
\end{lemma}
\begin{proof}
    Theorem \ref{thm:safety} ensures that $\underline{\by} \leq \by(\bx(t)) \leq \overline{\by}$ for all time, implying that $\by(\bx(t))$ is bounded for all time. Along with the assumption that $\by_{\rm{d}}(t)$ is bounded for all time, this implies that $\be(\bx(t),t)$ is bounded for all time, as desired.
\end{proof}
To further characterize tracking capabilities when using a safety filter, suppose there exists a locally Lipschitz controller $\bk_{\rm{d}}\,:\,\R^n\times\R_{\geq0}\rightarrow\R^m$ and a continuously differentiable function $V\,:\,\R^n\times\R_{\geq0}\rightarrow\R_{\geq0}$ such that:
\begin{equation}\label{eq:tracking-lyap}
    \rho\|\be(\bx,t)\|^2 \leq V(\bx,t),
\end{equation}
\begin{equation}\label{eq:tracking-kd}
    \begin{aligned}
        \pdv{V}{t}(\bx,t) + \pdv{V}{\bx}(\bx,t)\big[\bf(\bx) + \bg(\bx)\bk_{\rm{d}}(\bx,t) \big] \\ \leq - \gamma V(\bx,t) - \sigma\left\Vert\pdv{V}{\bx}(\bx,t)\bg(\bx)\right\Vert^2,
    \end{aligned}
\end{equation}
for all $(\bx,t)\in\R^n\times\R_{\geq0}$, where $\rho,\gamma,\sigma>0$. The existence of this controller and corresponding Lyapunov function $V$ encodes the ability of \eqref{eq:control-affine-system} to robustly track the desired output trajectory, in an input-state-stability (ISS) sense.
The following proposition illustrates that ISS properties of this tracking controller are preserved when used in the multi-CBF safety filter \eqref{eq:qp-multi-soln}. 
\begin{proposition}\label{prop:tracking}
    Let the conditions of Theorem \ref{thm:closed-form} hold and define $\bk(\bx, t)=\bu^*(\bx, t)$, with $\bu^*$ as in\footnote{While \eqref{eq:qp-multi-soln} was defined for time-invariant desired controllers, the expression can be modified appropriately to handle time-varying $\bk_{\rm{d}}$.} \eqref{eq:qp-multi-soln}, where the desired controller is chosen as $\bk_{\rm{d}}\,:\,\R^n\times\R_{\geq0}\rightarrow\R^m$ in \eqref{eq:tracking-kd} for a Lyapunov function $V\,:\,\R^n\times\R_{\geq0}\rightarrow\R_{\geq0}$ satisfying \eqref{eq:tracking-lyap} and \eqref{eq:tracking-kd}. If the safe set $\mc{S}$ is compact, $\bx_0\in\mc{S}$, and $\by_{\rm{d}}(t)$ and all of its derivatives are bounded, then, along the closed-loop system trajectories, the tracking error satisfies: 
    \begin{equation}\label{eq:tracking-bound}
        \|\be(\bx(t),t)\| \leq \beta(V(\bx(0),0), t) + \iota(\|\bk_{\rm{cbf}}(\bx(t), t)\|_{\infty}),
    \end{equation}
    for all $t\geq0$, where $\beta(r,s)\coloneqq\sqrt{\tfrac{r}{\rho}}e^{-\frac{\gamma}{2}s}$, $\iota(\mu)\coloneqq \tfrac{\mu}{2\sqrt{\gamma\rho\sigma}}$ and:
    \begin{equation}
        \bk_{\rm{cbf}}(\bx,t) \coloneqq \sum_{i=1}^m(\underline{\lambda}_i(\bx,t) - \overline{\lambda}_i(\bx,t))\bG^{-1}(\bx)\bb_i(\bx),
    \end{equation}
    with $\underline{\lambda_i},\overline{\lambda}_i$ defined analogously to \eqref{eq:qp-multi-soln}.
\end{proposition}
\begin{proof}
    First, note that $\bk(\bx,t)$ from \eqref{eq:qp-multi-soln} may be expressed as $\bk(\bx,t) = \bk_{\rm{d}}(\bx,t) + \bk_{\rm{cbf}}(\bx,t)$.  
    Hence, computing the derivative of $V$ the closed-loop system yields:
    \begin{equation*}
        \begin{aligned}
            \dot{V}(\bx,t)  = & 
            \pdv{V}{t}(\bx,t) + \pdv{V}{\bx}(\bx,t)\big[\bf(\bx) + \bg(\bx)\bk_{\rm{d}}(\bx,t) \big] \\
            & + \pdv{V}{\bx}(\bx,t)\bg(\bx)\bk_{\rm{cbf}}(\bx,t).
        \end{aligned}
    \end{equation*}
    Upper bounding the above using \eqref{eq:tracking-kd} yields:
    \begin{equation}
        \begin{aligned}
            \dot{V}(\bx,t) \leq & -\gamma V(\bx,t) - \sigma\left\Vert\pdv{V}{\bx}(\bx,t)\bg(\bx)\right\Vert^2 \\ &+ \left\Vert\pdv{V}{\bx}(\bx,t)\bg(\bx)\bk_{\rm{cbf}}(\bx,t)\right\Vert,
        \end{aligned}
    \end{equation}
    which, using Cauchy-Schwarz, can be further bounded as:
    \begin{equation}
        \begin{aligned}
            \dot{V}(\bx,t) \leq & -\gamma V(\bx,t) - \sigma\left\Vert\pdv{V}{\bx}(\bx,t)\bg(\bx)\right\Vert^2 \\ &+ \left\Vert\pdv{V}{\bx}(\bx,t)\bg(\bx)\right\Vert\cdot\|\bk_{\rm{cbf}}(\bx,t)\|.
        \end{aligned}
    \end{equation}
    Completing squares and further bounded then yields:
    \begin{equation}
        \dot{V}(\bx,t) \leq -\gamma V(\bx,t) + \frac{\|\bk_{\rm{cbf}}(\bx,t)\|^2}{4\sigma}.
    \end{equation}
    Let $t\mapsto\bx(t)$ be a trajectory of the closed-loop system, which satisfies $\bx(t)\in\mc{S}$ for all $t\geq0$ by Theorem \ref{thm:safety}. As $\mc{S}$ is compact, $\bx(t)$ is bounded. Further, since $\bk_{\rm{cbf}}$ is continuous and $\by_{\rm{d}}(t)$ and all of its derivatives are bounded, we have $\|\bk_{\rm{cbf}}(\bx(t),t) \|_{\infty}<\infty$.
    Hence, the derivative of $V$ along trajectories of the closed-loop system satisfies:
    \begin{equation}
        \begin{aligned}
            \dot{V}(\bx,t) \leq -\gamma V(\bx,t) + \frac{\|\bk_{\rm{cbf}}(\bx(t),t) \|_{\infty}^2}{4\sigma}.
        \end{aligned}
    \end{equation}
    Invoking the Comparison Lemma \cite[Lem. 3.4]{Khalil} implies:
    \begin{equation}
        V(\bx(t),t) \leq V(\bx(0),0)e^{-\gamma t} + \frac{\|\bk_{\rm{cbf}}(\bx(t),t) \|_{\infty}^2}{4\gamma\sigma},
    \end{equation}
    and rearranging terms using \eqref{eq:tracking-lyap} yields \eqref{eq:tracking-bound}, as desired.
\end{proof}

Proposition \ref{prop:tracking} illustrates that the tracking error is input-to-state stable with respect to intervention from the safety filter, viewed here as a disturbance to the nominal tracking objective. Establishing this result required assuming that $\mc{S}$ is compact. This will often be the case when $\sum_{i=1}^mr_i=n$ (i.e., when the output and its derivatives account for the entire state of the system), but may not hold in general. In such a case, one may impose the somewhat more restrictive assumption that $\bx(t)$ is bounded apriori to establish the same result. 

\section{Numerical Examples}\label{sec:sims}
\subsection{Planar Drone}\label{sec:planar-drone}
We illustrate our results on a planar drone with state $\bx=(x,z,\theta,\dot{x},\dot{z},\dot{\theta})\in\R^6$ capturing the horizontal position, vertical position, and orientation of the drone along with their velocities. The dynamics \eqref{eq:control-affine-system} of the system are described by:
\begin{equation}\label{eq:drone}
    \bf(\bx) = 
    \begin{bmatrix}
        \dot{x} \\ \dot{z} \\ \dot{\theta} \\ 0 \\ -g \\ 0
    \end{bmatrix},
    \quad 
    \bg(\bx) = 
    \begin{bmatrix}
        0 & 0 \\ 0 & 0 \\ 0 & 0 \\ -\sin(\theta) & 0 \\ \cos(\theta) & 0 \\ 0 & 1
    \end{bmatrix},
\end{equation}
where the input $\bu=(F,M)\in\R^2$ represents the total thrust and moment produced by the propellers. Our objective is to constrain the height of the quadrotor within the bounds $z\in[z_{\min}, z_{\max}]$, where $z_{\min} < z_{\max}$. We formalize this objective by defining the output $\by(\bx)=(z,\theta)$, where the need to include $\theta$ as an output variable will become clear shortly, along with the additional output constraints $\theta\in[\theta_{\min},\theta_{\max}]$, 
where $\theta_{\min}<\theta_{\max}$. 
The associated decoupling matrix is:
\begin{equation*}
    \bB(\bx) = \begin{bmatrix}
        \cos(\theta) & 0 \\ 
        0 & 1
    \end{bmatrix},
\end{equation*}
from which one may verify that $\by$ has vector relative degree $\br=(2,2)$ on the set $\mathcal{E}\coloneqq\R^2\times(-\tfrac{\pi}{2},\tfrac{\pi}{2})\times\R^3$. Hence, provided that $\theta_{\min}>-\tfrac{\pi}{2}$ and $\theta_{\max}<\tfrac{\pi}{2}$ we have:
\begin{equation*}
    \mc{C} = \{\bx\in\R^n\,:\,z_{\min}\leq z\leq z_{\max},\,\theta_{\min}\leq \theta\leq \theta_{\max}\},
\end{equation*}
as the overall state constraint set \eqref{eq:state-constraint-set}. Following the procedure from Sec. \ref{sec:main}, these output constraints are used to construct two pairs of second-order ECBFs with $\bm{\alpha}=\bm{\alpha}_1=\bm{\alpha}_2=(1,2)$ which produces a polynomial as in \eqref{eq:ECBF-poly} with real negative roots $\nu_1=\nu_2=-1$. Moreover, since $\by$ has a vector relative degree on $\mathcal{E}$, $z_{\max}>z_{\min}$, $\theta_{\max}>\theta_{\min}$, and $\alpha_1^1=\alpha_1^2=1>0$, Lemma \ref{lemma:ECBF-compat} ensures that the resulting ECBF constraints are compatible while Lemma \ref{lemma:lipschitz} ensures the corresponding QP-based controller \eqref{eq:ecbf-qp-multi-1} is locally Lipschitz continuous. Rather than numerically solving \eqref{eq:ecbf-qp-multi-1}, however, we leverage the closed-form solution from Theorem \ref{thm:closed-form} by selecting the weight in the quadratic cost as $\bG(\bx)=\bB(\bx)\T\bB(\bx)$. This safety filter is used to modify the inputs of a nominal controller that attempts to drive the drone to different setpoints outside of the constraint set.

The results of applying this safety filter to the drone are provided in Fig. \ref{fig:drone}, where the top left plot illustrates the evolution of $z(t)$, the top right plot illustrates the evolution of $\theta(t)$, the bottom left plot illustrates the control signals, and the bottom right plot illustrates the evolution of the ECBF constraints. As guaranteed by Theorem \ref{thm:safety} the outputs remain within their prescribed bounds, represented as the dashed black lines Fig. \ref{fig:drone}. Here, the desired set point, represented by the dotted red curve in the top left plot of Fig. \ref{fig:drone}, is periodically changed from $\bx=(1,1,\bzero_{4})$ to $\bx=(-1,0.2,\bzero_{4})$ so that the drone must safely change its orientation to move from one location to another. The satisfaction of the corresponding ECBF constraints (i.e., the left-hand side of the inequalities in \eqref{eq:ecbf-qp-multi-2}) is shown in Fig. \ref{fig:drone} (bottom right), where the values of the constraints remain positive for all time in accordance with the results of Theorem \ref{thm:closed-form}. 

\begin{figure}
    \centering
    \includegraphics{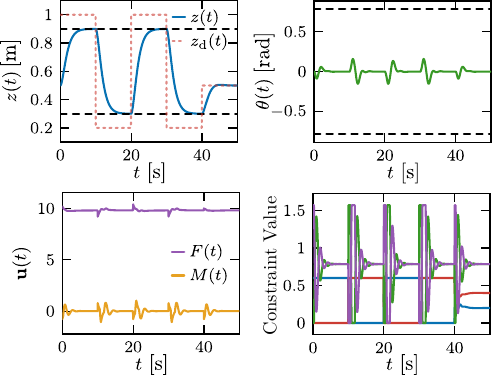}
    % \includegraphics{figures/pgf_drone_z.pdf}
    % \hfill
    % \includegraphics{figures/pgf_drone_th.pdf}
    % \hfill
    % \includegraphics{figures/pgf_drone_u.pdf}
    % \hfill
    % \includegraphics{figures/pgf_drone_con.pdf}
    \vspace{-0.2cm}
    \caption{Closed-loop trajectories of the drone example from Sec. \ref{sec:planar-drone} showing the drone's vertical position (top left), orientation (top right), inputs (bottom left), and ECBF constraints (bottom right).}
    \label{fig:drone}
\end{figure}

\subsection{More Complex Constraints via Reduced-Order Models}\label{sec:drone-rom}
One interesting aspect of the previous example is that, although the objective was to constrain $z$, achieving a relative degree required taking $\theta$ as an output. With only two inputs, this makes simultaneously constraining the drone's horizontal and vertical position challenging since the orientation must inevitably consume one of the outputs. Perhaps somewhat discouraging, then, is that the evolution of the orientation in Fig. \ref{fig:drone} comes nowhere near the critical states where the output's vector relative degree vanishes, raising if it is truly necessary to take the orientation as an output.

One solution to the issue presented above is to develop a CBF-based controller for a reduced-order model (RoM) of the drone, and then track these commands on the full-order model using a certified tracking controller (i.e., one with a corresponding Lyapunov function). A full theoretical characterization of this approach can be found in \cite{CohenARC24,CohenACC25} and is omitted here in the interest of space, but can be summarized as: the full-order model is safe if it tracks the commands generated by the RoM ``well enough." Inspired by the differential flatness properties of drones \cite{MellingerICRA2011}, we consider the two-dimensional double integrator with gravity:
\begin{equation}\label{eq:double-integrator}
    \dot{\bz} = \begin{bmatrix}
        \dot{x} \\ \dot{z} \\ 0 \\ -g
    \end{bmatrix}
    +
    \begin{bmatrix}
        0 & 0 \\ 0 & 0 \\ 1 & 0 \\ 0 & 1
    \end{bmatrix}
    \bv,
\end{equation}
with state $\bz=(x,x,\dot{x},\dot{z})\in\R^4$ and input $\bv\in\R^2$ representing the drone's acceleration as a RoM of \eqref{eq:drone}. The output $\by(\bz)=(x,z)$ has vector relative degree $\br=(2,2)$, implying that constraints of the form \eqref{eq:output-constraints} may be imposed on this system using CBFs to ensure that:
\begin{equation*}
    \bz(t)\in\mc{C} = \{\bz\in\R^4\,:\,x_{\min}\leq x \leq x_{\max},z_{\min}\leq z \leq z_{\max}\}
\end{equation*}
for all $t\geq0$. These constraints are used to construct ECBFs with the same parameters as in Sec. \ref{sec:planar-drone}, and yields a closed-form safety filter \eqref{eq:qp-multi-soln}, denoted by $\bk_{\bz}\,:\,\R^4\rightarrow\R^2$, for \eqref{eq:double-integrator}. The accelerations generated by the double integrator's safety filter are converted to commands for the drone \eqref{eq:drone} via \cite{BahatiCDC25}:
\begin{equation}
    \begin{aligned}
        F = & \begin{bmatrix}
            -\sin(\theta) \\ \cos(\theta)
        \end{bmatrix}
        ^{\dagger}\bk_{\bz}(x,z,\dot{x},\dot{z}), \\ 
        \theta_{\rm{d}} = & \atantwo\left(\bk_{\bz,1}(x,z,\dot{x},\dot{z}), \bk_{\bz,2}(x,z,\dot{x},\dot{z})\right),
    \end{aligned}
\end{equation}
where $(\cdot)^{\dagger}$ denotes the left pseudo-inverse, and $\theta_{\rm{d}}$ is a desired orientation that is tracked by a PD controller to produce the corresponding moment $M$. Provided this desired orientation is tracked ``well enough," then this approach ensures safety of \eqref{eq:drone}, with more technical details available in \cite{CohenACC25,BahatiCDC25}.

The results of applying this controller to the drone in \eqref{eq:drone} are illustrated in Fig. \ref{fig:drone-rom}. Here, the drone's trajectory (blue curve) is shown to stay inside the constraint set, indicated by the dashed black lines in Fig. \ref{fig:drone-rom} (top left). This is also reflected in the bottom right plot of Fig. \ref{fig:drone-rom}, where the values of the four ECBFs remain nonnegative for all time. Note that the desired setpoints for the drone (purple dots in Fig. \ref{fig:drone-rom} (top left)) are chosen such that the drone approaches the corners of the constraint set, where the ECBFs constraining the horizontal and vertical position are active simultaneously. This is indicated in Fig. \ref{fig:drone-rom} (top right), where two out of the four Lagrange multipliers from \eqref{eq:qp-multi-soln} are active on certain time intervals, highlighting the ability of the proposed approach to handle multiple CBF constraints simultaneously. 

\begin{figure}
    \centering
    \includegraphics{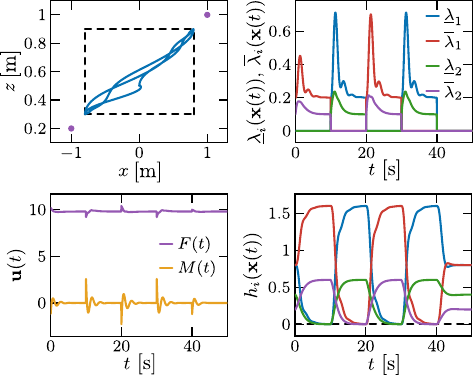}
    % \includegraphics{figures/pgf_drone_rom_xz.pdf}
    % \hfill
    % \includegraphics{figures/pgf_drone_rom_lambda.pdf}
    % \hfill
    % \includegraphics{figures/pgf_drone_rom_u.pdf}
    % \hfill
    % \includegraphics{figures/pgf_drone_rom_h.pdf}
    \vspace{-0.2cm}
    \caption{Closed-loop trajectories of the drone example from Sec. \ref{sec:drone-rom} showing the drone's $(x,z)$ position (top left), the Lagrange multipliers associated with the closed-form QP controller from \eqref{eq:qp-multi-soln} (top right), the control inputs (bottom left), and the ECBFs (bottom right).}
    \label{fig:drone-rom}
    % \vspace{-0.5cm}
\end{figure}

\section{Conclusions}\label{sec:conclusions}
This paper proposed a framework for the constrained control of nonlinear systems using CBFs. By focusing on the special case of box constraints on the system outputs, we derived conditions under which the CBFs enforcing these constraints are compatible and when the resulting QP controller admits a closed-form solution. We additionally presented results on how nominal tracking objectives degrade in the presence of CBF constraints. Future research directions include generalizing to broader classes of output constraints and incorporating input bounds.

\bibliographystyle{ieeetr}
\bibliography{biblio}

\section*{Acknowledgments}
We gratefully acknowledge Dr. Heather Hussain, whose continuous feedback on the ideas presented herein has greatly improved the quality of this work.

\end{document}